\documentclass[10pt,a4paper]{amsart}
\usepackage[left=3.5cm,right=3.5cm,top=4cm,bottom=4cm]{geometry}
\usepackage{setspace,color}
\usepackage{amsfonts, amssymb,tikz,pictex, latexsym, hyperref, amsthm,color,mathrsfs,graphicx, amsmath,caption,tabularx,verbatim,comment,pdfpages}

\newtheorem{thm}{Theorem}[section]

\newtheorem{lem}[thm]{Lemma}

\newtheorem{prop}[thm]{Proposition}

\newtheorem{claim}{Claim}

\newtheorem{thm-con}[thm]{Theorem-Conjecture}
\numberwithin{equation}{section}

\theoremstyle{definition}
\newtheorem{defn}[thm]{Definition}
\newtheorem{rmk}[thm]{Remark}
\allowdisplaybreaks

\newcommand{\F}{\mathbb F}
\newcommand{\fq}{\F_{\hskip-0.7mm q}}
\newcommand{\cfq}{\overline{\F}_{\hskip-0.7mm q}}

\newcommand{\f}{\Bbb F}
\newcommand{\ev}{\mathrm{ev}}
\newcommand{\Span}{\mathrm{Span}}

\makeatletter
\@ifpackageloaded{hyperref}%
{\newcommand{\mylabel}[2]
	{\protected@write\@auxout{}{\string\newlabel{#1}{{#2}{\thepage}%
				{\@currentlabelname}{\@currentHref}{}}}}}%
{\newcommand{\mylabel}[2]
	{\protected@write\@auxout{}{\string\newlabel{#1}{{#2}{\thepage}}}}}
\makeatother

\definecolor{Amaranto}{rgb}{0.9, 0.17, 0.31}
\definecolor{Borgogna}{rgb}{0.5, 0.0, 0.13}

\newcommand\blfootnote[1]{%
	\begingroup
	\renewcommand\thefootnote{}\footnote{#1}%
	\addtocounter{footnote}{-1}%
	\endgroup
}

\begin{document}
	
	\title{Codes from $A_m$-invariant polynomials}
	
\author[G. Micheli]{Giacomo Micheli}
\address{University of South Florida\\ 4202 E Fowler Ave\\
33620 Tampa, US.}
\email{gmicheli@usf.edu}
	
\author[V. Pallozzi Lavorante]{Vincenzo Pallozzi Lavorante}
\address{University of South Florida \\ 4202 E Fowler Ave\\
33620 Tampa, US.}
\email{vincenzop@usf.edu}

\author[P. Waitkevich]{Phillip Waitkevich}
\address{University of South Florida \\ 4202 E Fowler Ave\\
33620 Tampa, US.}
\email{phillipwaitkevich@usf.edu}

\blfootnote{This project is supported by NSF grant N: 2127742 and 2338424}

		\begin{abstract}
Let $q$ be a prime power.
This paper provides a new class of linear codes that arises from the action of the alternating group on $\mathbb F_q[x_1,\dots,x_m]$ combined with the ideas in  (M. Datta and T. Johnsen, 2022).
 Compared with Generalized Reed-Muller codes with analogous parameters, our codes have the same asymptotic relative distance but a better rate. Our results follow from combinations of  Galois theoretical methods with Weil-type bounds for hypersurfaces.
	\end{abstract}
	\begin{small}
		\maketitle
	{\bf Keywords:} Reed-Muller codes, Alternating group, permutations.
	
	{\bf 2000 MSC:} {11T71, 11T06, 13B05, 20B35}.
	
\end{small}

	\section{Introduction}

Let $q$ be a prime power, $\mathbb F_q$ be the finite field of order $q$, and $m$ be a positive integer. Constructing families of evaluation codes has always attracted a lot of interest due to the numerous applications to coding theory like error correction, DSS and SDMM \cite{d2020gasp,garrison2023class,hollanti2023algebraic,lopez2022secure,Micheli}.

Generalized Reed-Muller codes provide an extension of Reed-Solomon codes to the multivariate ring of polynomials. However, they have good relative distance (distance/length) but poor rate (dimension/length).
Thus, it is interesting to find sub-codes of Generalized Reed-Muller with the same asymptotic relative distance but a better rate.

Along this view, in \cite{datta2023codes}, Datta and Johnsen study a new class of codes that arises from the symmetric group. Such classes of codes have interesting parameters and the structural properties of the symmetric group allow them to derive important properties for the codes, such as the minimum distance or certain weight distribution properties for the generalized Hamming weight. Datta--Johnsen codes are essentially constructed by considering evaluations of linear combinations of elementary symmetric polynomials in a certain number of variables $m$. The minimum distance computation for such codes follows from the special factorization properties that these polynomials have, which in turn is a consequence of the fact that they are invariant under the symmetric group.
Let $A_m$ be the alternating group. This is an interesting general fact: whenever a class of multivariate polynomials in $F=\mathbb F_q[x_1,\dots,x_m]$ is invariant under a group action, then Galois theory over the fraction field of $F$ applies and leads to interesting properties for the factorization of such polynomials. In turn, this allows us to provide bounds for the number of zeroes of these polynomials, and therefore of certain codes constructed from these, as we will show in this paper for the case of $G=A_m$. Apart from providing a new general framework to construct codes from Galois theory, our paper provides advantages over Datta--Johnsen codes (which were already a significant improvement over Reed-Muller codes), since for a fixed $q$ we can construct codes with the same asymptotic rate and same relative distance but double length and dimension. Therefore, when codes are compared for a fixed finite field size, our codes have larger distance because we allow for more evaluation points, and also the message space can be extended (thanks to the fact that we are requiring polynomials to be invariant under a smaller subgroup).
The paper is structured as follows. In Subsection \ref{subsec:linco}, we recap the basic notions from the theory of linear error correcting codes. In Subsection \ref{subsec:points}, we include results that are needed to study the number of points on affine varieties. In Subsection \ref{subsec:lincombination}, we introduce the space of linear combinations of elementary symmetric polynomials and provide some properties from \cite{datta2023codes} that allow us to count the number of zeroes of polynomials in this space. In Subsection \ref{subsec:Galois}, we derive some properties of a certain set of polynomials in Lemma \ref{lm:factor}, that will be useful to determine the message space for our codes.
Section \ref{sec:boundze} is devoted to providing a bound on the number of zeroes of polynomials in our message space: this is done by splitting the proof into the two cases prescribed by Subsection \ref{subsec:linindep} and Subsection \ref{subsec:lindep}.
Finally, Section \ref{sec:constrcode} provides the construction of our codes and comparison with Datta-Johnsen codes and Reed-Muller codes for analogous parameters.
	
	\section{Background}
	\subsection{Linear codes}\label{subsec:linco}
	A code $ C $ of length $ n $ over the finite field $\f_q$ is a subset of $ \f_q^n $. The code $C$ is said to be linear of dimension $k$ if it is a $k$-dimensional $\f_q$-subspace of $\f_q^n$. The weight of an element of $ \f_q^n $ is defined to be the number of its non-zero entries. The \textit{Hamming distance} between two elements $ x, y $ of $ \f_q^n $ is defined to be the weight of $ x - y $. The \textit{minimum distance} $d$ of a code $ C $ is the minimum of distances between all two distinct elements of $ C,$ and
	by an $ [n,k,d]_q $ code we mean a linear code of length $n$, dimension (as a subspace) $k$ and minimum distance $d$.\\
	One may ask whether a code is a ``good" code compared to other constructions, this is why it is useful to introduce the notion of relative distance and  rate of a code. 
	\begin{defn}
		Let $C$ be a $[n,k,d]_q$ code. The relative distance is $\delta:=d/n$ and the rate is defined to be $\rho:=k/n$.
	\end{defn}
	We can compare linear codes for the same length by comparing their relative distance and rate. Codes with higher relative distance and/or rate are better than codes with lower ones.
	Generalized	Reed-Muller codes consist of the evaluation vectors of multivariate polynomials
	over $ \fq $. Let $\fq[x_1,\dots,x_m]$ be the polynomial ring with $m$ variables. The $t$-th order Generalized Reed-Muller code $GR_q(m,t)$ is defined as
	
	\begin{equation}\label{code:rm}
	GR_q(m,t):=\{(f(x) : x \in \fq^m) \mid  f\in \fq[x_1,\dots,x_m], \deg(f)\leq t\}
	\end{equation}

	and it is a $[q^m, \binom{m+t}{m}, (1-\frac{t}{q})q^m]_q$ code, see classic literature \cite{kasami1968new}.
	\

	\subsection{Points on varieties.}\label{subsec:points}
	Let $\cfq$ denote the algebraic
	closure of the field $\fq$. Let $F_1,\dots,F_\ell$ be polynomials in $\fq[x_1,\dots,x_m]$ and let $V$
	denote the affine subvariety of $\mathbb{A}^m(\cfq)$ defined by
	$F_1,\dots,F_\ell$. Counting or estimating the number of
	$\mathbb{F}_q$-rational points $x\in\mathbb{A}^m(\mathbb{F}_q)$ of $V$ is an important
	the subject of mathematics and computer science, with many
	applications. In \cite{cafure2006improved} the authors showed that the number $|V(\fq)|$ of $\mathbb{F}_q$-rational points of an
	$\fq$-absolutely irreducible hypersurface $V$ of $\mathbb{A}^m(\cfq)$
	of degree $\delta>0$ is:
	\begin{equation}\label{estimate:Ghorpade-Lachaud}||V(\fq)|-q^{m-1}| \leq
		(\delta-1)(\delta-2)q^{m-3/2}+5\delta^{13/3}q^{m-2}.
	\end{equation}
	For more details see \cite[Theorem 5.2]{cafure2006improved}.	
	In the next section, we will use this result to bound the number of zeros of certain polynomial equations.
	
	\subsection{The vector space of elementary symmetric polynomials}\label{subsec:lincombination}
	In \cite{datta2023codes} the authors studied the vector space generated by the elementary symmetric polynomials in $m$ variables. We recall here some useful properties that will be needed in the next sections.
	We denote by $\sigma_m^i$ the $i$-th elementary symmetric polynomial in $m$ variables $x_1, \dots, x_m$, i.e.,
	$$\sigma_{m}^i = \sum_{1 \le j_1 < \cdots < j_i \le m} x_{j_1} \cdots x_{j_i}$$
	for $1 \le i \le m$ and $\sigma_m^0 = 1$. 
	The following result is obtained by collecting the results in \cite[Section 2]{datta2023codes}
	\begin{prop}\label{pr:s}
		Let $s \in \fq[x_1, \dots, x_m]$ be given by	$s = a_0 + a_1 \sigma_m^1 + \dots + a_m \sigma_m^m$	where $a_0, \dots, a_m \in \fq$. Then	$s$ is either absolutely irreducible, say of type $1$, or $s = a \displaystyle{\prod_{i=1}^m} (\alpha + x_i)$ for $a,\alpha \in \fq$, say of type $2$.
	\end{prop}
	
	\begin{rmk}\label{rk:sn-1}
		Note that, given a polynomial $s$ that is a linear combination of elementary symmetric polynomials, by isolating one variable, say $x_1$, we can write $s=x_1 p_1 +p_2$, where $p_1$ and $p_2$ are linear combination of elementary symmetric polynomials in $x_2,\dots,x_m$ (hence invariant under the action of $S_{m-1}$). 
	\end{rmk}
	
	\subsection{Galois theory and $A_m$-invariant polynomials}\label{subsec:Galois}
	Let $A_m$ be the alternating group of $m$ variables, that is the subgroup of $S_m$ of all the even permutations. $A_m$ acts on the set of polynomials $\overline{\F}_q[x_1,\dots,x_m]$ by acting on its variables. More specifically, if $\sigma \in A_m$, then $f(x_1,\dots,x_m)$ is sent to $\sigma(f):=f(x_{\sigma(1)},\dots, x_{\sigma(m)})$
	\begin{defn}
		An $A_m$-invariant polynomial $f\in \overline{\F}_q[x_1,\dots,x_m]$ is a polynomial that is invariant under the action of $A_m$, that is $f=\sigma(f)$ for every $\sigma \in A_m$.
	\end{defn}
	Note that, in particular, any symmetric polynomial is $A_m$-invariant.
	The following result is classical and will be used later in the paper. We include the proof for completeness.
	
	\begin{thm}\label{th:index}
		Let $A_m$ be the alternating group. Then it does not have a proper subgroup of index less than $m$, for $m\geq5$.
	\end{thm}
	\begin{proof}
		Assume $ A_m $	has a subgroup $ G $ of index $ m'<m $. Then the action of $A_m$ on the cosets of $ G $
		gives a homomorphism into $ S_{m'} $. Since $ m\geq5 $, $ m!/2>m'! $, so the homomorphism can't be injective. Since $ A_m $
		is simple, the kernel must be all of $ A_m $. In particular, this means that $ hG=G $ for all $h \in A_m$, which is only possible if $ G=A_m$. Thus, there is no proper subgroup of index less than $ m $.
	\end{proof}
	
	Let \[v_m(x)=\prod_{1\leq i <j\leq m}(x_i-x_j)\] be the Vandermonde polynomial in $m$ variables. $v_m$ is invariant under every even permutation, while every odd permutation results in a change of sign.  This means that $v_m$ is an $A_m$-invariant polynomial that is not symmetric.
	The following is a well-known property of $A_m$-invariant polynomials. We provide a short proof using Galois theory for completeness.
	\begin{lem}\label{lm:factor}
		Let $g \in \overline{\F}_q[x_1,\dots,x_m]$ be an $A_m$-invariant polynomial. Then there exist $s_1, s_2 \in \overline{\F}_q[x_1,\dots,x_m]$ symmetric polynomials such that: \[g=s_1 v_m + s_2,\] for $v_m$ being the Vandermonde polynomial in $m$ variables. Furthermore, the representation is unique.
	\end{lem}
	\begin{proof}
		We know that $[\overline{\F}_q(x_1,\dots,x_m)^{A_m}:\overline{\F}_q(x_1,\dots,x_m)^{S_m}]=2$ since the index of $A_m$ in $S_m$ is $2$, where $\overline{\F}_q(x_1,\dots,x_m)^{A_m}$ and $\overline{\F}_q(x_1,\dots,x_m)^{S_m}$ are the fixed fields of $A_m$ and $S_m$ respectively. Thus, by the fundamental theorem of Galois Theory (and the fact that every polynomial that is invariant under the symmetric group is an algebraic combination of elementary symmetric polynomials), the field of rational functions invariant under $A_m$ can be written as $\overline{\F}_q(\sigma_m^1,\dots,\sigma_m^m,v_m)$, where $\sigma_m^1,\dots,\sigma_m^m$ are the elementary symmetric polynomials in $m$ variables and $v_m$ is the Vandermonde polynomial in $m$ variables.
		In particular any $A_m$-invariant polynomial $h\in \overline{\F}_q[x_1,\dots,x_m]$ uniquely decomposes as follows:
		$h(x)=\frac{p_1}{p_2}+v_m\frac{p_3}{p_4}$, 
		for $p_1,p_2,p_3,p_4$ being symmetric polynomials and $\gcd(p_1,p_2)=\gcd(p_3,p_4)=1$. This means that 
		\[h(x)=\frac{p_1p_4+v_mp_3p_2}{p_2p_4}.\]
		Since $p_2p_4$ is symmetric and for an odd permutation $\sigma$ we have $\sigma(p_1p_4+v_mp_3p_2)=p_1p_4-v_mp_3p_2$ (because $v_m$ is simply the square root of the discriminant in $T$ of $\prod^m_{i=1}(T-x_i)$), we get that $p_2p_4 | p_1p_4$ and $p_2p_4 | v_m p_3p_2$. Thus $p_2 |p_1$ and $p_4 | p_3$, prove that the rational functions are polynomials (note that $p_4 \nmid v_m$ by the definition of the decomposition).
	\end{proof}
	
	\begin{rmk}\label{rk:schur}
	It is well known that the set of degree $d$ Schur polynomials in $m$ variables are a linear basis (over $ \fq $) for the space of homogeneous degree $d$ symmetric polynomials in $m$ variables. This implies that every symmetric polynomial is a sum of homogeneous symmetric polynomials. Thus, if in the decomposition of Lemma \ref{lm:factor} the symmetric polynomial $s_1$ is different from $0$, then $g$ must have a total degree at least $\binom{m}{2}$. Let $s_1=p + \tilde{s}_1$ where $p\ne 0$ is the leading degree homogeneous polynomial, then $v_m p$ is a homogeneous alternating polynomial of degree $\deg(p)+\binom{m}{2}$ and it cannot be canceled with any term of $s_2$. 
	\end{rmk}

	\section{Bound for the number of zeros}\label{sec:boundze}
	Let $x=(x_1,\dots,x_m) \in \mathbb{A}^m(\fq)$, for $q$ odd. Consider the following polynomial:
	\begin{equation}\label{eq:F(x)}
		F(x):=s_1(x)v_m(x)+s_2(x) \in \fq[x],
	\end{equation}
	for $s_1(x)$ and $s_2(x)$ being linear combinations of elementary symmetric polynomials and $v_m$ being the Vandermonde polynomial in $m$ variables.
	\begin{rmk}\label{rk:s1es2}
		Note that $s_1$ and $s_2$ are either linearly dependent or they cannot share any common components. In fact by {\rm Proposition \ref{pr:s}}, $s_1$ and $s_2$ are either both absolutely irreducible or both of type $2$. Thus if they are both of type $2$ and they share one component, they need to be $\fq$-linearly dependent, i.e. scalar multiples, (simply because sharing a factor ensures that they share all factors).
	\end{rmk}
	We are interested in computing the number of zeros of a polynomial of the form \eqref{eq:F(x)}. More specifically we want to compute the number $|Z_D(F)|$ of the distinguished zeros of $F$, where a point $(a_1, \dots, a_m) \in \mathbb{A}^m (\fq)$ is said to be \textit{distinguished} if $a_i \neq a_j$ whenever $i \neq j$. Note that the set of distinguished points of $\mathbb{A}^m(\fq)$, say $\mathbb{A}^m_D(\fq)$, has cardinality $|\mathbb{A}^m_D(\fq)|=P(q,m)$, where 
	$$
	P(q,m)=
	\begin{cases}
		{q \choose m} m! \ \ \ \text{if} \ \ m \le q,\\
		0 \ \ \ \ \ \text{otherwise}.
	\end{cases}
	$$
	We now state the main theorem of this paper that will allow us to give a lower bound for the distance of our codes.
	\begin{thm}\label{th:main}
		Let $F$ be a polynomial as in equation \eqref{eq:F(x)} and let $d:=\gcd(\binom{m}{2},q-1)$. Then for $q\geq m^{10}$ and $m\geq6$ we have \[ |Z_{D}(F)| \leq \frac{P(q,m)}{q-1}d+ mP(q-1,m-1). \]
		
	\end{thm}
	
	In the following, we will distinguish when $s_1$ and $s_2$ are linearly dependent or not and treat those two cases separately.
	
	\subsection{Linearly independent case}\label{subsec:linindep}
	
	The set of distinguished zeros of $F$ can be computed as follows. Let $Z(f)$, $Z_D(f)$ and $Z_{ND}(f)$ be the set of zeros, \textit{distinguished} zeros and \textit{non-distinguished} zeros of a polynomial $f$, respectively. Since for every non-distinguished zero $z=(z_1,\dots,z_m)$ of $F$ we have $v_m(z)=0$, the following holds:
	\begin{equation}\label{eq:bound}
		\begin{aligned}
			|Z_D(F)| =|Z(F)|-|Z_{ND}(F)|=|Z(F)|-|Z_{ND}(s_2)|=|Z(F)|-\left(|Z(s_2)|-|Z_D(s_2)| \right),
		\end{aligned}
	\end{equation}  
	that is $	|Z_D(F)|= |Z(F)|-|Z(s_2)|+|Z_D(s_2)|$.
	
	\begin{lem}
		Let $m \geq 6$. If $s_1$ and $s_2$ are linearly independent, the polynomial $F$ defined by equation \eqref{eq:F(x)} is absolutely irreducible.
	\end{lem}
	\begin{proof}
		Let $g \in \cfq[x_1,\dots,x_m]$ be a divisor of $F$. We may suppose $g$ is absolutely irreducible.  Since $F$ is not symmetric (or otherwise $s_1=0$, and $s_1,s_2$ are linearly dependent), then it cannot split only into symmetric irreducible factors, hence we may assume that $g$ is not symmetric.
		
		Since $F$ is stabilized by the alternating group $A_m$, any polynomial $\sigma(g)$, for $\sigma \in A_m$, is a factor of $F$. Let $G$ be the stabilizer of $g$ in $A_m$. We have two cases:
		\begin{itemize}
			\item $G= A_m$. In this case, $g$ is fixed by $A_m$. By Lemma \ref{lm:factor} $g$ can be written as $g=r_1 v_m + r_2$ where $r_1$ and $r_2$ are symmetric polynomials. Thus we have \begin{equation} \label{eq:uq1}
				(r_1 v_m + r_2)\ell = s_1 v_m +s_2,
			\end{equation} 
			for $\ell \in \overline{\F}_q[x_1,\dots,x_m]$. 	Since $g$ and $F$ are $A_m$-invariant, then $\ell$ is also stabilized by $A_m$, and we can write $\ell=t_1 v_m +t_2$ for $t_1,t_2$ symmetric polynomials.
			Note that $r_1\ne0$ since $g$ is not symmetric. Moreover, $r_2\ne0$ since $g$ is irreducible. Finally, $t_2\ne 0$ since $v_m \nmid F$ (or otherwise $s_2=0$, denying the linear independence). In particular since $r_1\ne 0$ Remark \ref{rk:schur} implies that $\deg(g)\geq \binom{m}{2}$.

	The latter forces $\deg(\ell) \leq m$ (since $\deg(F)\leq \binom{m}{2}+m$) and in turn, $t_1=0$, by Remark \ref{rk:schur}.

			Hence $\ell$ is symmetric. By the uniqueness of the representation applied to the equation \eqref{eq:uq1}, we obtain that $r_1 \ell =s_1$ and $r_2\ell =s_2$, which is in contradiction with $s_1$ and $s_2$ being linearly independent (from Remark \ref{rk:s1es2}).
			\item  $G < A_m$. 
			
			\begin{claim}\label{cl:orbit}
	       Let $m\geq 6$. Then $\deg_{x_i}(g)\leq 1$ for all $i \in \{1,\dots,m\}$. 
			\end{claim}
			
			\begin{proof}[Proof of claim \ref{cl:orbit}]By Theorem \ref{th:index}, $G$ has index at least $m$.
				This means that the orbit of $g$ under the action of $A_m$ has cardinality at least $m$, by the orbit-stabilizer theorem. Consider now the degree of $F$ in the variable $x_i$, say $\deg_{x_i}(F)$. We have that $\deg_{x_i}(F)\leq m$.
				If every variable appears in $g$, then it must be that $\deg_{x_i}(g)\leq 1$; in fact, each factor in the product obtained by acting $A_m$ on $g$ contains all the variables and we have at least $m$ factors.
				Let $g$ be without exactly one variable, say $i^*$. Every element in $G$ must be in $St_{A_m}(i^*) \simeq A_{m-1}$, where $St_{A_m}(i^*)$ is the stabilizer of $i^*$ in $A_m$. By applying again Theorem \ref{th:index}, we derive that $G$ has index at least $m-1$ in $A_{m-1}$ which implies that the index of $G$ in $A_m$ is at least $m(m-1)$. Let $1$ be the index such that $\deg_{x_1}(g)\geq2$.
				Note that $\deg_{x_1}(F)\geq[St_{A_m}(1):St_G(1)]\deg_{x_1}(g)$, because $F$ is invariant under $A_m$ and different representations of the cosets of $St_G(1)$ in $St_{A_m}(1)$ move $g$ to a different factor of $F$ with same degree in $x_1$. Using the orbit-stabilizer theorem we derive:
				\[[St_{A_m}(1):St_G(1)]=\frac{|A_{m-1}|}{|St_G(1)|}= \frac{(m-1)!}{2}\frac{|Or_G(1)|}{|G|},\]
				where $Or_G(1)$ is the orbit of $1$ under the action of $G$. Since the index of $G$ in $A_m$ is at least $m(m-1)$, then $|G|\leq \frac{m!}{2m(m-1)}$. This implies that
				\[[St_{A_m}(1):St_G(1)]\geq (m-1)|Or_G(1)|\geq m-1.\]			
				The latter implies $\deg_{x_1}(F)\geq 2m-2$, a contradiction.
                Finally, let $g$ be without $2$ or more variables, say $x_1$ and $x_2$. In this setting, we note that there are at least $2 \binom{m-2}{2}$ elements in the orbit of $g$ under $A_m$. We have at least the even permutations of the following form: $(1\quad i)(2\quad j)$ for $i,j\in \{3,\dots,m\}$, $i\ne j$. Since $m \geq 6$ and $\deg_{x_i}(F)\leq m$, we get a contradiction.
			\end{proof}
			Now if we consider the reduction modulo $g$, we get that:
			\begin{equation}\label{eq:modg}
				s_1v_m\equiv-s_2 \pmod{g}
			\end{equation}
			Let us now exclude that $s_2 \equiv 0 \pmod{g}$. First, observe that $g \neq c s_2$ for any $c \in \fq$ because otherwise it would be fixed by $A_m$. Therefore if $g$ were to divide $s_2$, we would have $s_2$ reducible and $g=\alpha+x_i$ for some $\alpha\in \fq$ and $i \in \{1,\dots,m\}$ by Remark \ref{rk:s1es2}. Since $g$ cannot divide $v_m$ this implies that $g$ divides $s_1$ and in turn, this forces $s_1$ and $s_2$ to be linearly dependent by Remark \ref{rk:s1es2}, a contradiction.
			Thus, $s_2 \not\equiv 0 \pmod{g}$. 
Without loss of generality we can suppose $\deg_{x_1}(g)=1$ and  $\deg_{x_i}(g)\leq 1$ for $i \in \{2,\dots,m\}$. We isolate $x_1$ from $g$ in the quotient ring $ \overline{\F}_q[x_1,\dots,x_m]/(g)$ obtaining $x_1\equiv \frac{h_1}{h_2}$ in $ \overline{\F}_q[x_1,\dots,x_m]/(g) $ (in other words, there is an natural isomorphism $ \overline{\F}_q[x_1,\dots,x_m]/(g) \rightarrow \overline{\F}_q[h_1/h_2,x_2,\dots,x_m]$), for some $h_1,h_2 \in \overline{\F}_q[x_2,\dots,x_m]$ such that $\deg_{x_i}(h_2)\leq 1$, and  $\deg_{x_i}(h_1)\leq 1$ for $i \in \{2,\dots,n\}$, and coprime.
			By Remark \ref{rk:sn-1}, we can write $s_1=x_1 p_1+p_2$ and $s_2=x_1 r_1+r_2$ where $p_1,p_2,r_1,r_2$ are linear combination of symmetric elementary polynomials in $x_2,\dots,x_m$. 
Therefore, since $\overline{\F}_q[x_1,\dots,x_n]/(g)$ 
can be embedded in $\overline{\F}_q(x_2,\dots x_n)$ thanks to the fact that the degree of $g$ in $x_1$ is $1$, Equation \eqref{eq:modg} becomes
\[\left(\frac{h_1}{h_2}p_1+p_2\right)\left(\frac{1}{h_2^{m-1}}\right)\left(\prod_{i=2}^m(h_1-h_2 x_i)\right)v_{m-1}=-\left(\frac{h_1}{h_2}r_1+r_2\right).\]
By multiplying both sides by $h_2^{m}$ we get
			\begin{equation}\label{eq:maineqlemma}
   \left(h_1p_1+h_2p_2\right)\left(\prod_{i=2}^m(h_1-h_2 x_i)\right)v_{m-1}=-h_2^{m-1}h_1r_1-h_2^{m}r_2.
   \end{equation}
Suppose that $h_2$ is not constant. Then, $h_2$ has an irreducible factor, say $u$. Now, at least $u^{m-1}$ divides the RHS above. The LHS, on the other hand, cannot be divisible by  $u^{m-1}$ for $m\geq 4$ as we now explain. Recall that $h_2$ is coprime to $h_1$. The factor $v_{m-1}$ is squarefree (so at most one power of $u$ divides it), the product in $i$ is coprime to $h_2$ (so no powers of $u$ can divide it), if $h_1p_1+h_2p_2$ is divisible by $u$ then $p_1$ is divisible by at least $u^{m-2}$ (which is a contradiction because factorizations of linear combinations of elementary symmetric polynomials are squarefree, as prescribed by Proposition \ref{pr:s}).

For the case in which $h_2$ is constant, it is enough to check the total degree of both sides of \eqref{eq:maineqlemma}. In fact, the RHS has total degree at most $2m-2$, while the LHS has total degree at least $(m-1)(m-2)/2+m$, a contradiction for $m\geq 6$.

		\end{itemize}

	\end{proof}
	
	Thanks to the previous lemma, we can use equation \eqref{estimate:Ghorpade-Lachaud} to bound $Z(F)$ and $Z(s_2)$. We have that $\deg (F)\leq\binom{m}{2}+m\leq m^2$ (for $m\geq 2$) and $\deg(s_2)\leq m$, hence
	\[|Z(F)| \leq q^{m-1}+ (m^2-1)(m^2-2)q^{m-3/2}+5m^{26/3}q^{m-2}\]
	and
	\[|Z(s_2)| \geq q^{m-1}- (m-1)(m-2)q^{m-3/2}-5m^{13/3}q^{m-2}. \]
	Note that we do not need $s_2$ irreducible to obtain the correspondent bound since if $s_2$ is reducible we can lower bound the number of zeros of any of its irreducible components, still obtaining a lower bound for the zeros of $s_2$ (and then we can upper bound the degree of its irreducible component with $m$, as it appears with negative sign).
	This implies that for $m\geq 4$
	\[|Z_D(F)| \leq m^4 q^{m-3/2}+ m^2 q^{m-3/2}+5m^{26/3}q^{m-2}+5m^{13/3}q^{m-2}+Z_D(s_2).\]
	
	In \cite{datta2023codes} the authors provided a sharp bound for the number of distinguished zeros of a symmetric polynomial obtained as a linear combination of elementary symmetric polynomials, that is 
	\[|Z_D(s_2)| \leq m{P}(q-1,m-1),\]
	which implies that 
	\begin{equation}\label{eq:b1}
		|Z_D(F)| \leq m^4 q^{m-3/2}+ m^2 q^{m-3/2}+5m^{26/3}q^{m-2}+5m^{13/3}q^{m-2}+m!\binom{q-1}{m-1},
	\end{equation}
	since 
	$$
	P(q,m)=
	\begin{cases}
		{q \choose m} m! \ \ \ \text{if} \ \ m \le q,\\
		0 \ \ \ \ \ \text{otherwise}.
	\end{cases}
	$$
	\subsection{Linearly dependent case}\label{subsec:lindep}
	Let $M:=\binom{m}{2}$ and $\gcd(M,q-1) = 1.$ Let $\mathbb{A}^m_D(\F_q)$ be the set of all distinguished points of $\F_q^m$, i.e. points with non-repeated coordinates in $ \fq $. We will show in this section that if $s_1$ and $s_2$ are linearly dependent, then   
	\[ |Z_{D}(F)| \leq \frac{|\mathbb{A}^m_D(\F_q)|}{q-1}+ mP(q-1,m-1). \] We begin with a few necessary lemmas for the proof of the above claim. 
	
	\begin{rmk}\label{lem:characterizeAD}
		Note that $x \in \mathbb{A}^m_D(\F_q)$ if and only if $v_m(x) \ne 0,$ and $v_m$ is surjective. In fact, $v_m(\lambda x) = \lambda^Mv_m(x)$ and the map  $\iota: \F_q^* \mapsto \F_q^*$ given by $\iota(x) = x^M$ is a bijection since we are assuming $\gcd(M,q-1)=1$.
	\end{rmk} 
	Our next goal is to show that there are two orthogonal partitions of $\mathbb{A}^m_{D}(\fq)$.
	
	\begin{lem}
		Let $\mathcal{P}_1$ be the partition determined by the pre-images of $v_m$. For every $z \in \mathbb{A}^m_D(\F_q)$, let $B_z :=  \{cz : c \in \F_q^* \}.$ Then the collection of sets $\mathcal{P}_2:=\{B_z : z \in \mathbb{A}^m_D(\F_q) \}$ is a partition of $\mathbb{A}^m_D(\F_q)$. In particular, $\mathcal{P}_1$ and $\mathcal{P}_2$  are orthogonal partitions and  $| v_m^{-1}(\lambda)|=|\mathcal{P}_2|$. 
	\end{lem}
	\begin{proof}
		Note that either $B_x \cap B_y= \emptyset $ or $B_x=B_y$. In fact, there exists $ z \in B_x \cap B_y$ if and only if $z= \lambda_1 x= \lambda_2 y$, for non-zero elements $\lambda_1$ and $\lambda_2$, which implies that $x=\lambda_2/\lambda_1 y$, or equivalently, $B_x=B_y$. Hence $\mathcal{P}_2$ is a partition.
		
		Now it remains to show the orthogonality of the two partitions. Let $\lambda \in \fq^* $ and $x \in v_m^{-1}(\lambda)$. By definition, $ x \in B_x  \in \mathcal{P}_2$. For every $y=\lambda_1 x \in B_x$ with $\lambda_1\in \fq$ we obtain that if \[v_m(y)=\lambda \iff \lambda_1^M v_m(x)=\lambda \iff \lambda_1^M \lambda=\lambda  \iff \lambda_1=1 \iff x=y,\]
		since $\gcd(M,q-1)=1$.
		Thus, each element $x \in v_m^{-1}(\lambda)$ belongs to a unique set	 $B_x \in \mathcal{P}_2$, showing that the two partitions are orthogonal and that $| v_m^{-1}(\lambda)|=|\mathcal{P}_2|$. 
	\end{proof}

	\begin{thm} \label{thm: depcasegcd1}
		Let $M:=\binom{m}{2}$ and $\gcd(M,q-1)=1$. Let $F$ be a polynomial of the form given in equation \eqref{eq:F(x)}. If $s_1$ and $s_2$ are linearly dependent, we have that   
		\[ |Z_{D}(F)| \leq \frac{P(m,q)}{q-1}+ mP(q-1,m-1). \]
	\end{thm}
	\begin{proof}
		Suppose that $s_1$ and $s_2$ are dependent. Then, $s_2 = \lambda s_1$ for some $\lambda \in \f_q^*.$ Hence, we can write $F = s_1(x)v_m(x) + s_2(x) = s_1(x)v_m(x) + \lambda s_1(x) = (v_m(x)+\lambda)(s_1(x)).$ We have from \cite{datta2023codes} that $Z_D(s_1) \leq mP(q-1,m-1),$ and so it remains to show a bound for the distinguished zeroes of $v_m(x)+\lambda.$ Observe that this is the same as finding the largest set in $\{|v_m^{-1}(c)| : c\in \F_q^*\}$ since $v_m(x) + \lambda = 0 \iff v_m(x) = -\lambda$. By the above lemma, we know that $| v_m^{-1}(\lambda)|=|\mathcal{P}_2|$ for every $\lambda \in \fq^*$. Observe that each $B_z \in \mathcal{P}_2$ covers $q-1$ distinct points in $\mathbb{A}^m_D(\F_q)$. So, $|\mathcal{P}_2| = \frac{|\mathbb{A}^m_D(\F_q)|}{q-1}$. Hence, we have that $|v_m^{-1}(\lambda)| = \frac{|\mathbb{A}^m_D(\F_q)|}{q-1}$ for every $\lambda \in \fq^*$, that is $v_m(x) = \lambda$ on exactly $\frac{|\mathbb{A}^m_D(\F_q)|}{q-1}$ many points. In conclusion, $|Z_D(F)| \leq |Z_D(v_m+\lambda)|+ |Z_D(s_1)|  = \frac{|\mathbb{A}^m_D(\F_q)|}{q-1}+ mP(q-1,m-1).$ 
	\end{proof}
	The case for $ \gcd(M,q-1) > 1$ ($M:=\binom{m}{2}$) is more complicated. We cannot use anymore that the map $\iota(x)=x^M$ is a bijection. This is why the bound on the number of zeros of $v_m+\lambda$ for $\lambda\ne0$ is not sharp anymore. However, by using another argument we were still able to prove a generalization of the previous bound also for $\gcd(M,q-1) > 1$, which we decided to separate from the Theorem \ref{thm: depcasegcd1}, which is instead sharp.
	\begin{thm}\label{thm:gcdnot1}
		Let $M:=\binom{m}{2}$ and $d := \gcd(M,q-1) > 1$. Let $F$ be a polynomial of the form given in equation \eqref{eq:F(x)}. If $s_1$ and $s_2$ are linearly dependent, we have that   
		\begin{equation}\label{eq:b2}
			|Z_{D}(F)| \leq \frac{P(q,m)}{q-1}d+ mP(q-1,m-1). 
		\end{equation}
	\end{thm}
	\begin{proof}
        As in Theorem \ref{thm: depcasegcd1} it is only needed to show a bound for the distinguished zeroes of $v_m(x)+\lambda.$
		Let $\lambda \in \F_q^*$. Observe that there are $d$ solutions in $\F_q^*$ to the equation $\lambda^M = 1$; in fact, if $\xi$ is a primitive element of $\F_q^*,$ then the set $S = \{1,\xi^\frac{q-1}{d},\xi^\frac{2(q-1)}{d}, \cdots, \xi^\frac{(d-1)(q-1)}{d}\}$ is the set of the solutions to the latter equation. This means that for any $x \in v_m^{-1}(\lambda)$, the elements $\xi^\frac{q-1}{d}x,\xi^\frac{2(q-1)}{d}x, \cdots, \xi^\frac{(d-1)(q-1)}{d}x$ are also in $v_m^{-1}(\lambda).$ Denote by $S_x$ the set $\{sx : s\in S\},$ and let $B_x = \{\lambda x : \lambda \in \F_q^*\}$. As we saw before, each $B_x$ covers $q-1$ distinct elements and $S_x \subset B_x$.
		Let $x,y \in v_m^{-1}(\lambda)$ such that $y \notin S_x$. We claim that $B_x \cap B_y = \emptyset$. In fact if there were $\lambda_x, \lambda_y \in \fq$ such that $\lambda_y y =\lambda_x x$, then \[\lambda_y y =\lambda_x x \implies \lambda_y^M v_m(y)=\lambda_x^M v_m(x) \implies \lambda_y^M=\lambda_x^M \implies \frac{\lambda_x}{\lambda_y} \in S \mbox{ and } \frac{\lambda_y}{\lambda_x} y = x,\] which is in contradiction with $y \notin S_x$.
		
		Finally observe that there are at most $t := \frac{|\mathbb{A}^m_D(\F_q)|}{q-1}$ distinct points $z_1, z_2, \dots, z_t \in v_m^{-1}(\lambda)$ such that $B_{z_1}, B_{z_2},\dots, B_{z_t}$ are all disjoint; in fact each set contains $q-1$ distinct points in $\mathbb{A}^m_D(\F_q)$ and $|\bigcup_{i=1}^t B_{z_i}| = t(q-1) = |\mathbb{A}^m_D(\F_q)|$.  Since for each of those $z_i$'s there are $d$ elements in $v_m^{-1}(\lambda)$ (corresponding to the elements in $S_{z_i}$), we derive that $|v_m^{-1}(\lambda)| \leq td = \frac{|\mathbb{A}^m_D(\F_q)|}{q-1}d$. Now we conclude as in the proof of Theorem \ref{thm: depcasegcd1}.
	\end{proof}

	\begin{proof}[Proof of the main Theorem \ref{th:main}]
	We obtained the following bounds respectively for the linear independent case and linearly dependent case:
	\[|Z_D(F)| \leq m^4 q^{m-3/2}+ m^2 q^{m-3/2}+5m^{26/3}q^{m-2}+5m^{13/3}q^{m-2}+m!\binom{q-1}{m-1},\]
	and
	\[	|Z_{D}(F)| \leq \frac{P(q,m)}{q-1}d+ m!\binom{q-1}{m-1}. \]
	By comparing the different terms of the two equations, that is 
\[\frac{P(q,m)}{q-1}d=q(q-2)\cdots(q-m+1)d, \mbox{ and } m^4 q^{m-3/2}+ m^2 q^{m-3/2}+5m^{26/3}q^{m-2}+5m^{13/3}q^{m-2}, \]
we derive that for $q\geq m^{10}$ and $m\geq6$, we have to take the bound of \eqref{eq:b2}. Thus, we obtain the claim since the RHS of both bounds are increasing functions in $m$ and the bound \eqref{eq:b2} is asymptotically larger.
	\end{proof}

\begin{rmk}
It is out of the scopes of this paper to work out the cases $m\leq 6$, or $q<m^{10}$ which is a relevant but technical task, which we leave to the interested reader.
\end{rmk}

	\section{Construction of Codes from $A_m$-invariant Polynomials} \label{sec:constrcode}
	\subsection{Construction} In this last section we show how to construct linear codes from $A_m$-invariant polynomials. Let $m \in \mathbb{N}$ be large enough such that $\gcd(m,q-1)=1$, let $\sigma_m^i$ the $i$-th elementary symmetric polynomial in $m$ variables and let 
	
	\begin{equation}\label{code:C}
	\Sigma_m := \left\{ s_1 + v_m s_2 : s_1=\sum_{i=0}^ma_i \sigma_{m}^i, \, s_2=\sum_{i=0}^mb_i \sigma_{m}^i, \, a_i,b_i \in \fq, \, \forall i \in \{0,\dots,m\} \right\}.
	\end{equation}

	  Let $\mathbb{A}^m_D(\F_q)$ be the set of all distinguished points in $\F_q^m.$ Consider the group action $\phi: A_m \times \mathbb{A}^m_D(\F_q) \to \mathbb{A}^m_D(\F_q)$ defined by $\phi(\sigma,P)=P_\sigma$, where if $P=(x_1,\dots,x_m)$ then $P_\sigma:=(x_{\sigma(1)},\dots, x_{\sigma(m)})$.
	The points of $\mathbb{A}^m_D(\F_q) $ constitute a disjoint union of orbits under the action $\phi$, and each orbit has cardinality $m!/2$. Thus, we can define a code by evaluating the polynomials in $\Sigma_m$ on a smaller evaluation set, consisting of one point from each of the $A_m$ orbits mentioned before.
	Let $n=2{q \choose m}$, and let $P_1,\dots,P_n$ be a set of representatives, one from each orbit.
Consider the evaluation map $\ev: \Sigma_m \to \F_q^n$ given by \[\ev(F) := (F(P_1),F(P_2),\dots, F(P_n)).\] Then, we define $C := \ev(\Sigma_m).$
\begin{prop}\label{prop:C}
	For $q\geq m^{10}$ and $m\geq6$, $C$ is a linear code with length $n = 2{q \choose m},$ dimension $k = 2(m+1),$ and distance $d \geq n - \left(2\frac{{q \choose m}}{q-1} + 2{q-1 \choose m-1}\right)$. 
\end{prop}
	 
	\begin{proof}
		The length of $C$ equals the number of orbits of $\mathbb{A}^m_D(\F_q)$ under the action of $A_m$. Note that $|\mathbb{A}^m_D(\F_q)| = P(q,m),$ and that we partitioned $|\mathbb{A}^m_D(\F_q)|$ using orbits of size $\frac{m!}{2}.$ So, the number of orbits is $2 {q \choose m}.$ Hence, $n = 2 {q \choose m}.$ Now, we show that $k = 2(m+1).$ Consider the set $S = \{\sigma_m^0,\sigma_m^1,\dots, \sigma_m^m\}$ where $\sigma_m^i$ is the $i^{th}$ symmetric polynomial in $m$ variables. In \cite{datta2023codes} it is shown that the elements in $S$ are linearly independent. Observe that $v_m S := \{ v_m s : s \in S\}$ is a $\fq$-linearly independent set of $m+1$ polynomials. Since we have $\Span \{S\}\cap \Span\{v_mS\}=0$, then $\Sigma_m = \Span\{ S \} \oplus \Span \{v_m S\},$ and this is a vector space of dimension $2(m+1).$ Finally, let $F_{max} \in \Sigma_m$ be such that $\displaystyle{|Z_D(F_{max})| = \max_{f \in \Sigma_m}|Z_D(f)|}.$ Observe that just like $\mathbb{A}^m_D(\F_q),$ $Z_D(F_{max})$ can be partitioned by orbits of size $m!/2$, and so the maximum number of coordinates equal to $0$ that a codeword could have is $2|Z_D(F_{max})|/m!$. Hence by Theorem \ref{thm: depcasegcd1}, \[d = n - \frac{2|Z_D(F_{max})|}{m!} \geq n - \left(2\frac{{q \choose m}}{q-1} + 2{q-1 \choose m-1}\right).\]
	\end{proof}
	
	\begin{rmk}
	    
	    Even if our result relies on the Hasse-Weil theorem for large values of $q$, using Sage \cite{sagemath}, it is easy to check that our codes maintain the same parameters also for small values of $q$, provided that $q\geq m-1 \geq 5$.
        The reason is that the bound obtained for the linearly dependent case does not require any asymptotic assumption and that is the case when the set of zeros for our family of polynomials has the largest cardinality. 
	\end{rmk}

	\subsection{Asymptotic Comparisons with other codes}
	In this subsection, we investigate the relative distance $\delta_C$ and rate $\rho_C$ our code $C$ described in Proposition \ref{prop:C} by comparing it to the closest (in terms of regime of parameters) available constructions. In particular, our codes and Datta-Johnsen codes achieve better asymptotic parameters than Generalized Reed-Muller codes.
\subsubsection{Datta-Johnsen codes from symmetric polynomials}	In \cite{datta2023codes}, the authors constructed a code $C^\prime$ with length $n^\prime = {q \choose m}$, dimension $k^\prime = m+1$, and distance $d^\prime = {q \choose m} - {q-1 \choose m-1}$. The length and dimension of $C$ are twice the length and dimension of $C^\prime$, respectively. It can be shown that for fixed $m$ the relative distance of $C$ and $C^\prime$ are asymptotically equal as $q$ grows. That is, \[\lim_{q \to \infty} \frac{\delta_C}{\delta_{C'}}=1.\]
These considerations imply that for for a fixed $q$ and the same information rate, our codes have double the distance.
\subsubsection{Generalized Reed-Muller codes}	In addition to that, it makes sense to compare our code to the Generalized Reed-Muller code \eqref{code:rm} for $t=m$, where $t$ is the degree of the polynomials and $m$ is the number of variables. In this case, we observe that while we get asymptotically the same relative distance, our code $C$ provides asymptotically a better rate; for example, for $q$ being the next prime power after $m^{10}$, $\rho_{RM}\sim \binom{2m}{m}/(m^{10m})$ and $\rho_C\sim m/\binom{m^{10}}{m}$, and
	\[\lim_{m\to \infty}\frac{\rho_C}{\rho_{RM}}=\infty\]

	\section{Future Work} It should be possible to extend the ideas used in this paper and \cite{datta2023codes} to create codes from arbitrary subgroups of $S_m$ (the symmetric group of $m$ variables). We briefly outline the strategy. Let $x_1, x_2, \dots, x_m$ be variables and let $H$ be a subgroup of size $N$ of the symmetric group $S_m$. Let $K = \F_q(s_1,s_2,\dots, s_m)$ where $s_i$ represents the $i^{th}$ elementary symmetric polynomial. Let $L = \F_q(x_1,x_2,\dots,x_m)$. Denote $L^H$ as the set of polynomials in $L$ fixed by $H.$ By the fundamental theorem of Galois Theory, the degree of the field extension $L^H / K$ is equal to $|H| = N.$ By the definition of degree of a field extension, this means that $\exists f_1, f_2, \dots, f_N \in L^H$ such that $L^H = f_1K + f_2K + \cdots +f_NK.$ We can construct linear codes similarly to how we proceed in this paper: let $H$ act on the set $\mathbb{A}^m_D(\F_q)$ and create codewords by evaluating a polynomial in $L^H$ at a distinct representative of each orbit. Their length $n$ should be $N {q \choose m}$, dimension $N(m+1),$ and distance is expected to be roughly $\displaystyle{n - \frac{N}{m!}\max_{f \in L^H}Z_D(f) }$. 

 Another question is whether it is possible to improve the bound of in Theorem  \ref{thm:gcdnot1} (the bound in Theorem \ref{thm: depcasegcd1} is instead sharp).
 
Finally, it would be very interesting to improve the bounds at the end of Section \ref{subsec:linindep} by using geometric properties of the varieties arising in the counting argument. In particular, Theorem \ref{th:main} only gives a regime of parameters in which our codes are guaranteed to exist: it would be very interesting to see if it is possible to relax the conditions on $q$ and $m$ with more advanced counting techniques.

	\bibliographystyle{abbrv}
	\bibliography{bib_sym.bib}

\end{document}